\documentclass[11pt, oneside]{amsart}   	

\usepackage{amsmath,amssymb,amsthm}
\usepackage{psfrag,graphics,epsfig,epic,color,subfigure}
\usepackage{graphicx}      
\usepackage{epstopdf} 
\usepackage{algorithmic}
\usepackage{algorithm}
\usepackage{comment}
\usepackage{epstopdf,subfigure}
\usepackage{url}
\usepackage{booktabs}

\newtheorem{Theorem}{Theorem}

\newtheorem{Lemma}{Lemma}

\newtheorem{Definition}{Definition}

\newcommand{\R}{\mathbb{R}}

\newcommand{\x}{\mathbf{x}}

\newcommand{\p}{\mathbf{p}}

\newcommand{\e}{\mathbf{e}}

\newcommand{\E}{\mathbb{E}}

\begin{document}

\title{APTER: Aggregated Prognosis Through Exponential Reweighting\footnotemark[2]}

\author{K.~Pelckmans, L. Yang\\
	Division of Systems and Control,\\
	Department of Information Technology, \\
	Uppsala University, Sweden}

\begin{abstract}
	This paper considers the task of learning how to make a prognosis of a patient based on his/her micro-array expression levels.
	The method is an application of the aggregation method as recently proposed in the literature on theoretical machine learning, 
	and excels in its computational convenience and capability to deal with high-dimensional data.
	A formal analysis of the method is given, yielding rates of convergence similar to what traditional techniques obtain,	
	while it is shown to cope well with an exponentially large set of features. 
	Those results are supported by numerical simulations on a range of  publicly available survival-micro-array datasets.
	It is empirically found that the proposed technique combined with a recently proposed 
	preprocessing technique gives excellent performances.
\end{abstract}


\maketitle

\footnotetext[2]{	The used software files and datasets are available on \url{http://www.it.uu.se/research/project/survlab}}

\section{Introduction}

Learning how to make a prognosis of a patient is an important ingredient to the 
task of building an automatic system for personalised medical treatment.
A prognosis here is understood as a useful characterisation of the (future) time of an event of interest.
In cancer research, a typical event is the relapse of a patient after receiving treatment.
The traditional approach to process observed event times is addressed in the analysis of survival data, 
see e.g. \cite{kalbfleisch2011statistical} for an excellent review of this mature field in statistics.
Most of those techniques are based on parametric or semi-parametric assumptions on how the data was generated.

Probably the most prevalent technique is Cox' Proportional Hazard (PH) approach, 
where inference is made by maximising a suitable partial likelihood function.
This approach has proven to be very powerful in many applications of survival analysis, 
but it is not clear that the basic assumption underlying this technique 
holds in the analysis of micro-array datasets. 
Specifically, the proportional hazard assumption is hard to verify and might not even be valid.
This in turn jeopardises the interpretation of the results.
This is especially so since the data has typically a high dimensionality while typically a few (complete) cases are available, 
incurring problems of ill-conditioning.
Many authors suggested fixes to this problem, see for examples \cite{tibshirani2009univariate} and references.
Some of such work proposed in the early 2000, was studied numerically and compared in \cite{bovelstad2007predicting}.
In applied work, one often resorts to a proper form of preprocessing in order to use Cox' PH model, see e.g. \cite{van2002gene}.

Since prognosis involves essentially a form of prediction, it is naturally to phrase this problem in a context of modern machine learning.
This insight allowed a few authors to come up with algorithms which are deviating from a likelihood-based approach.
We mention here \cite{van2011learning} and references therein.

This work takes this route even further.
It studies the question {\em how can new insights in machine learning help to build a more powerful algorithm}?
As dictated by the application, we are especially interested in dealing with high-dimensional data. 
That is, cases where many ($O(10^4)$) covariates might potentially be relevant,
while only relatively few cases ($O(10^2)$) are available.
Furthermore, we are not so much interested in {\em recovering} the mechanisms underlying the data
since that is probably too ambitious a goal. Instead, we merely aim at making a good {\em prognosis}.
It is this rationale that makes the present technique essentially different from likelihood-based, or penalised likelihood-based approaches 
as e.g. the PH-L$_1$ \cite{goeman2010l1,simon2011regularization} or the Danzig Selector for survival analysis \cite{antoniadis2010dantzig}, 
and points us resolutely to methods of machine learning and empirical risk minimisation.

The contribution of this work is threefold.
Firstly, discussion of the application of prognosis leads us to formulate a criterion which does 
not resort to a standard approach of classification, function approximation or maximum (partial) likelihood inference.
Secondly, we point to the use of aggregation methods in a context of bio-informatics,
give a subsequent algorithm (APTER) and derive a competitive performance guarantee.
Thirdly, we present empirical evidence which supports the theoretical insights, and affirms its use for the analysis 
of micro-array data for survival analysis. 
The experiments can be reproduced using the software made public at \url{http://www.it.uu.se/research/project/survlab}.

\subsection{Organization and Notation}

This paper is organized as follows.
The next section discusses the setting of survival analyses and the aim of prognosis. 
Section 3 describes and analyses the proposed algorithm.
Section 4 gives empirical results of this algorithms on artificial and micro-array datasets.
Section 4 concludes with a number of open questions.

This paper follows the notational convention to represent deterministic single quantities as lower-case letters,
vectors are denoted in bold-face, and random quantities are represented as upper-case letters.
Expectation with respect to any random variable in the expression is denoted as $\E$.
The shorthand notation $\E_n[\cdot]$ denotes expectation with respect to all $n$ samples seen thus far,
while $\E_{n-1}[\cdot]$ denotes expectation with respect to the first $n-1$ samples.
$\E^n[\cdot]$ denotes expectation with respect to the $n$th sample only, 
such that the rules of probability imply that $\E_n[\cdot] = \E_{n-1}\E^n[\cdot]$.

The data is represented as a set of size $n$ of tuples 
\begin{equation}
	\{(\mathbf{x}_i,Y_i,\delta_i)\}_{i=1}^n,
	\label{eq.data}
\end{equation}
Let $0<Y_1\leq Y_2 \leq \dots \leq Y_n$ be an ordered sequence of observed event times associated to $n$ subjects.
An event can be either a failure with time $T_i$, or a (right) censoringtime $C_i$, expressed as the time elapse from $t_0$.
In this paper we assume that all $n$ subjects share the same time of origin $t_0$.
It will be convenient to assume that each subject has a failure and right censoring time with values $T_i$ and $C_i$ respectively.
Then only the minimum time can be observed, or $Y_i = \min(T_i,C_i)$.
It will be convenient to define the {\em past event set} 
$P(t)\subset\{1,\dots,n\}$ at time $t$.
That is, $P(t)$ denotes the set of all subjects which have experienced an {\em event} strictly before time $t$.
Let for $i=1,\dots,n$ the indicator $\delta_i\in\{0,1\}$ denote wether the event (failure) is directly observed 
($\delta_i=1$), or if the subject $i$ is censored ($\delta_i=0$), or $\delta_i = I(Y_i< C_i)$. 
Then
\begin{equation}
	P(t) = \left\{i : \ Y_i< t, \delta_i=1\right\}.
	\label{eq.R}
\end{equation}
Furthermore, associate to each subject $i=1,\dots, n$ a covariate $\x_i\in\R^d$ of dimension $d$.
In the present setting, $d=O(1000)$, while $n=O(100)$ at best.

\section{Prognosis in Survival Analysis}

In this section we formalize the task of learning how to make a prognosis, based on observed cases. 
The general task of prognosis in survival analysis can be phrased as follows:
\begin{Definition}[Prognosis]
	Given a subject with covariate $\x_\ast\in\R^d$, what can we say about the value of its associated $T_\ast$?
\end{Definition}
Motivated by the popular essay by S.J. Gould\footnote{'The Median Isn't the Message' 
as in \url{http://www.prognosis.org/what_does_it_mean.php}}, we like to make statements as 
'my covariates indicate that with high probability I will outlive 50\% of the subjects suffering the same disease',
or stated more humanely as 'my covariates indicate that I belong to the {\em good} half of the people having this disease'.
The rationale is that this problem statement appears easier to infer than estimating the full conditional hazard or conditional survival functions,
while it is more informative than single median survival rates. 

Specifically, we look for an {\em expert} $f:\R^d\rightarrow\R$ which can decide for any 2 different subjects $0<i,j\leq n$ 
which one of them will {\em fail} first.
In other words, we look for  an $f$ such that for as many couples $(i,j)$ as possible, one has 
\begin{equation}
	(T_i - T_j) \left( f(\x_i) - f(\x_j)\right) \geq 0.
	\label{eq.ci}
\end{equation}
Since $T_k$ is not observed in general due to censoring, the following (rescaled) proxy is used instead
\begin{equation}
	\sum_{i=1}^n \frac{1}{|P(Y_i)|} \sum_{j\in P(Y_i)} I\left(f(\x_i)  <  f(\x_j)\right),
	\label{eq.ci}
\end{equation}
where $I(z)=1$ if $z$ holds true, and equals zero otherwise.
In case $|P(Y_i)|=0$, the $i$th summand in the sum is omitted. 
This is standard practice in all subsequent formulae.
Note that this quantity is similar to the so called Concordance Index ($C_n$) as proposed by Harell \cite{gonen2005concordance}.
The purpose of this paper is to propose and analyze an algorithm for finding 
such $f$ from a large set $\{f\}$, based on observations and under the requirements imposed by the specific setup.

If given one expert $f:\R^d\rightarrow\R$, its 'loss'  of a prognosis of a 
subject with covariate $\x_\ast\in\R^d$ and time of event $Y_\ast$ would be 
\begin{equation}
	\ell_\ast(f) =  \frac{1}{|P(Y_\ast)|}\sum_{k\in P(Y_\ast)} I\left(f(\x_\ast) \leq f(\x_k)\right).
	\label{eq.lossi}
\end{equation}
That is, $\ell_\ast(f)$ is the fraction of samples which experience an event before the time $Y_\ast$
associated to the subject with the covariate $\x_\ast$, although they were prognozed with a higher score by expert $f$.
Now we consider having $m$ experts $\{f_i\}_{i=1}^m$, and we will learn which of them performs best.
We represent this using a vector $\p\in\R^m$ with $\p_i\geq 0$ for all $i=1,\dots,m$, and with $1_m^T\p=1$.
Then, we will use this {\em weighting} of the experts to make an informed prognosis of the event to occur at 
$T_\ast$, of a subject with covariate $\x_\ast\in\R^d$.
Its associated loss is given as 
\begin{equation}{\small
	\ell_\ast(\p) = \sum_{i=1}^m \p_i \left( \frac{1}{|P(T_\ast)|}\sum_{k\in P(T_\ast)} I\left(f_i(\x_\ast) \leq f_i(\x_k)\right) \right).
	\label{eq.lossa}
}\end{equation}
This represents basically which expert is assigned most value to for making a prognosis.
For example, in lung-cancer we may expect that an expert based on smoking behaviour of a patient has a high weight.
Note that we include the $'='$ case in (\ref{eq.lossa}) in order to avoid the trivial cases where $f$ is constant.
So, we have formalised the setting as learning such $\p$ in a way that the smallest possible loss $\ell_\ast(\p)$
will be (or can be expected to be) made \footnote{
Note that different censoring distributions will have a different impact on the of this simple accuracy measure.
However, without too much effort one can compensate for that as in \cite{koziol2009concordance} by a proper weighting scheme. 
Since this is not essential to the technique per se, we omit that to the current manuscript.}.

\section{The APTER algorithm}

When using a {\em fixed} vector $\hat{\p}$, 
we are interested in the {\em expected} loss of the rule.
The expected loss of the $n$th sample $(\x_n,T_n)$ becomes 
$\mathbb{L}(\hat\p) = \E^n\ell_n(\hat\p) =$
\begin{equation}{\tiny
	\E^n\left[ \sum_{i=1}^m \hat\p_i\frac{1}{|P(T_n)|} \sum_{k\in P(T_n)} I \left(f_i(\x_n) \leq f_i(\x_k)\right)\right].
	\label{eq.hatLn}
}\end{equation}
Note that bounds will be given for this quantity which are valid for any $\x_n\in\R^d$ which may be provided.
In order to device a method which guarantees properties of this quantity, we use the mirror averaging algorithm 
as studied in A. Tsybakov, P. Rigollet, A. Juditsky in \cite{juditsky2008learning}.
This algorithm is based on ideas set out in \cite{NY1983}.
It is a highly interesting result of those authors that the resulting estimate has better
properties in terms of oracle inequalities compared to techniques based on sample averages.
Presently, such fast rate is not obtained since the involved loss function is not exponentially concave 
as in \cite{juditsky2008learning}, Definition 4.1.
Instead of this property, we resort to use of Hoeffding's inequality which gives us 
a result with rate $O(\sqrt{\frac{\ln m}{n}})$.
In order to give a formal guarantee of the algorithm, the following property is needed:

\begin{algorithm}
\caption{APTER: Aggregate Prognosis Through Exponential Reweighting}
\label{alg.apter}
\begin{algorithmic}
{\small
\STATE (0) Let $\p_i^0 = \frac{1}{m}$ for $i=1,\dots,m$, and fix $\nu>0$.
\FORALL {$k=1,\dots,n$}
\STATE (1)
	The prognosis associated to the $m$ experts $\{f_i\}_{i=1}^m$ are scored whenever {\em any} 
	new event (censored or not) is recorded for a subject $k\in\{1,\dots, n\}$ at time $Y_k$ as 
	\begin{equation}
		\ell_k(f_i) 
		= 
		\frac{1}{|P(Y_k)|}\sum_{l\in P(Y_k)} I\left(f_i(\x_k) \leq  f_i(\x_l) \right)
		\label{eq.apter.1}
	\end{equation}
 	and the cumulative loss is $L_k(f_i) = \sum_{s=1}^k \ell_s(f_i)$.
 
\STATE (2) 
	The vector $\p^k$ is computed for $i=1,\dots,m$ as follows
	\begin{equation}
		\p_i^k
		= \frac{\exp(-\nu L_k(f_i))}{\sum_{j=1}^m \exp(-\nu L_k(f_j))}.
		\label{eq.apter.2}
	\end{equation}
\ENDFOR
\STATE (3) Aggregate the hypothesis $\{\p^k\}_k$ into $\hat\p$ as follows:
	\begin{equation}
		\hat\p = \frac{1}{n} \sum_{k=0}^{n-1} \p^k.
		\label{eq.aggregation}
	\end{equation}
}
\end{algorithmic}
\end{algorithm}

\begin{Definition}
	For any $t=1,\dots,n$ and $i=1,\dots,m$ we have that
	\begin{equation}
		\E_n\left[ g\left(\frac{L_n(f_i)}{n}\right)\right] = \E_n[g( \ell_t(f_i))].
		\label{eq.cal}
	\end{equation}
	for any regular function $g:\R\rightarrow\R$.
\end{Definition}
This essentially means that we do not expect the loss to be different when it is measured at different points in time (different subjects).
\begin{Theorem}[APTER]
	Given $m$ experts $\{f_i\}_{i=1}^m$, and the loss function $\ell$ as defined in eq. (\ref{eq.lossa}).
	Then run the APTER algorithm with $\nu = \sqrt{\frac{2\ln m}{n}}$ resulting in $\hat\p$.
	Then
	\begin{equation}
			\E_{n-1}\left[
				\mathbb{L} (\hat\p)
				- 
				\min_{i=1,\dots,m} \mathbb{L}(f_i)   
		\right]
		\leq 
		\sqrt{\frac{2\ln m}{n}}.
		\label{eq.loss2}
	\end{equation}
\end{Theorem}

This result is in some way surprising.
It says that we can get competitive performance guarantees 
without a need for explicitly (numerically) optimizing the performance over a set of hypothesis.
Note that an optimization formulation lies on the basis of a maximum (partial) likelihood method
or a risk minimization technique as commonly employed in a machine learning setting.
There is an implicit link with optimization and aggregation through the method of mirror descent, see e.g. \cite{jakobson2000low} and \cite{bickel2006regularization}.
The lack of an explicit optimization stage results in the considerable computational speedups.
Note further that the performance guarantee degrades only as $\sqrt{\log(m)}$ in terms of the number of experts $m$.

\subsection{Choice of Experts and APTER$_p$}

The following experts are used in the application in micro-array case studies.
Here, we use simple univariate rules. 
That is, the experts are based on individual features (gene expression levels) of the dataset.
The rationale is that a single gene expression might well be indicative for the observed behaviours.

Let $\e_i$ be the $i$th unit vector, and let $\pm$ denote both the positive as well as the negated version.
Then, the experts $\{f_i\}$ are computed as 
\[ f_i(\x) = \pm \e_i^T\x, \]
so that $m=2d$, and every gene expression level can both be used for {\em over-expression} or {\em under-expression}.

In practice however, evidence is found that the following set of experts result in better performance:
\[f_i(\x) = s_i \e_i^T\x,\]
where the sign $s_i\in\{-1,1\}$ is given by wether the $i$th expression has a concordance index with the observed outcome larger or equal to $0.5$,  as estimated on the set used for training. This means that $m=d$. This technique is referred to as APTER$_p$.
Note that this subtlety needs also to be addressed in the application of Boosting methods.
There,  another popular choice is the use of random trees, see e.g. \cite{van2009feature}.

\subsection{Preprocessing using SIS and ISIS}

It is found empirically that preprocessing using the Iterative Sure Independence Screening ISIS 
as described in \cite{fan2008sure} improves the numerical results.
However, the rational for this technique comes from an entirely different angle.
That is, it is conceived as a screening technique for PH-L$_1$-type of algorithms.

The screening rule works as follows.
Let $\mathbf{m}=(\mathbf{m}_1,...\mathbf{m}_d)^T\in\R^d$ be defined as  
\begin{equation}
	\mathbf{m}=\sum_{i=1}^{n}Y_{i}\mathbf{x}_{i}.
	\label{eq.m}
\end{equation}
For any given $\gamma\in(0,1)$, define the set $M_\gamma$ as \cite{fan2008sure}:
\begin{equation}
M_{\gamma}=\left\{ 1\leq i\leq d:\left|\mathbf{m}_{i}\right|\textrm{is among the first \ensuremath{[\gamma n]} largest entries of } \mathbf{m} \right\}.
\end{equation}
Here, $[\gamma n]$ denotes the integer part of $\gamma n$. 
This set then gives the indices of the features which are retained in the further analysis. It is referred to as Sure Independence Screening (SIS) \cite{fan2008sure}.
In the second step, APTER is applied using only the retained features. 
Note that in the paper \cite{fan2008sure}, one suggests 
instead using a Cox partial Likelihood approach with a SCAD penalty
(for numerical comparison with such scheme, see the next section).

An extension of SIS is Iterative SIS (ISIS), see \cite{fan2008sure}. 
The idea is to pick up important features, missed by SIS.
This goes as follows: rather than having a single preprocessing (SIS) step, the procedure is repeated as follows.
At the end of a SIS-APTER step, a new (semi-) response vector $Y'$ can be computed by application of the found regression coefficients.
This new response variables can then be reused in a SIS step, resulting in fresh $[\gamma n]$ features.
This procedure is repeated until one has enough {\em distinct} features.

Since $[\gamma n]$ features are then given as input to the actual training procedure, 
we will refer to this value as $m$ in the experiments, making this connection between screening and training more explicit.

\section{Empirical Results}

This section present empirical results supporting the claim of efficiency.
First, we describe the setup of the experiments.

\subsection{Setup}

The following measure of quality (the Concordance index or $C_n$ or $C$-index, see e.g \cite{raykar2007ranking})
of a prognostic index scored by the function $f:\R^d\rightarrow\R$ is used.
Let again the data be denoted as $\{(\mathbf{x}_i,Y_i,\delta_i)\}_{i=1}^n$, where 
${\mathbf{x}_i}$ are the covariates, $Y_i$ contains the survival- or censoring time, and $\delta_i$ is the censoring indicator as before.
Consider any $f:\R^d\rightarrow\R$, then $C_n$ is defined as 
\begin{equation}
C_{n}(f)=\frac{\sum_{i:\delta_i=1}\sum_{Y_{j}>Y_{i}}I(f(\mathbf{x}_{i})<f(\mathbf{x}_{j}))}{\left|\varepsilon\right|}.
\end{equation}
Here $\left|\varepsilon\right|$ denotes the number of the pairs which
have $Y_{i}<Y_{j}$ when $Y_{i}$ is not censored. The indicator function
$I(\pi)=1$ if $\pi$ holds, and equals 0 otherwise.
That is, if $C_{n}(f)=1$, one has that $f$ scores a higher prognostic index to the subject with will experience the event later ('good').
A $C_{n}(f)=0.5$ says that the prognostic index given by $f$ is arbitrary with respect to event times ('bad').
Observe that this measure is not quite the same as $\ell_n(f)$ or $L_n(f)$ as were used in the design of the APTER algorithm.
Note that this function goes along the lines of the Area under the ROC curve or the Mann-Whitney statistic, adapted to handling censored data.

The data is assigned randomly to training data of size $n_t=\lfloor2n/3\rfloor$ and test data of size $n-n_t$.
The training data is used to follow the training procedures, resulting in $\hat{f}$. 
The test data is used to compute the performance expressed as $C_n(\hat{f})$.
The results are randomised 50 times (i.e. a random assignments to training and test set), 
and we report the median value as well as $\pm$ the variance.
The parameter $\nu>0$ is tuned in the experiments 
using cross-validation on the dataset which is used for training. 
It was found that proper tuning of this parameter is crucial for achieving good performance.

The following ten algorithms are run on each of these datasets:
\begin{itemize}
\item[(a)] APTER: The approach as given in Alg. \ref{alg.apter} where experts $\{f_i, f_i'\}$ are taken as $f_i(\x) = \e_i^T\x$ 
				and $f_i'(\x) = -\e_i^T\x$. In this way we can incorporate positive effects due to over-expression and under-expression
				of a gene. This means that $m=2d$.
\item[(b)] APTER$_p$: The approach as given in \ref{alg.apter} where experts $\{f_i\}$ are given as $f_i(\x) = s_i \e_i^T\x$
				where the sign $s_i\in\{-1,1\}$ is given by the $C_n$ of the $i$th expression with the observed effect, estimated on the 
				set used for training. This means that $m=d$.
\item[(c)] MINLIP$_p$:  
				The approach based on ERM and $s_i$ as discussed in \cite{van2011improved}.
\item[(d)] MODEL2:  	Another approach based on ERM as discussed in \cite{van2011improved}.
\item[(g)] PLS:		An approach based on preprocessing the data using PLS and application of Cox regression, as described in \cite{bovelstad2007predicting}.
\item[(f)] PH-L$_1$:    An approach based on a $L_1$ penalized version of Cox regression, as described in \cite{goeman2010l1}.
\item[(g)] PH-L$_2$:	An approach based on a $L_2$ penalized version of Cox regression, as described in \cite{goeman2008penalized}.
\item[(h)] ISIS-APTER$_p$:	An approach which uses ISIS as preprocessing, and applies APTER$_p$ on the resulting features \cite{fan2008sure}.
\item[(i)] ISIS-SCAD:	An approach which uses ISIS as preprocessing, and applies SCAD on the resulting features \cite{fan2008sure}.
\item[(j)] Rankboost:	An approach based on boosting the c-index \cite{freund2003efficient}.
\end{itemize}

Those algorithms are applied to an artificial dataset (as described below)
as well as on a host of real-world datasets (as can be found on the website).
Those datasets are publicly available, and all experiments can be reproduced using the code available at 
\footnote{The software is available at \url{http://www.it.uu.se/research/project/survlab}}.

\subsection{Artificial Data}

The technique is tested on artificial data which was generated as follows.
A disjunct training set and test set, both of size $100$ 'patients' was generated.
For each 'patient', $d$ features are sampled randomly from a standard distribution, so that $\x_i\in\R^d$.

We say that we have only $k$ {\em informative} features when
an event occurs at time $T_i$ computed for $i=1, \dots,n$ as 
\begin{equation}
	T_{i}=\frac{-\log Z_i}{10\exp\left(\sum_{j=1}^k \x_{i,j}\right)},
\end{equation}
where $Z_i$ is a random value generated from a uniform distribution on the unit interval $]0,1[$, 
and $\x_{i,j}$ is the $j$th covariate for the $i$th patient.
The right-censoring time is randomly generated from the exponential distribution with rate $0.1$.
After application of the censoring rule, we arrive at the {\em survival} time $Y_i$.

In a first experiment, $d$ is fixed as 100, but only the first $k\leq d$ features have an effect on the outcome ('informative').
Figure (\ref{fig:cerror}.a)  shows the evolution of the performance ($C_n(\hat{f})$) for increasing values of $k$.
In a second experiment we fix $k=10$, and record the performance for increasing values of $d$,
investigating the effect of a growing number of {\em ambient} dimension on the performance of APTER. 
Results are displayed in Figure (\ref{fig:cerror}.b). 

Thirdly, we investigate how well the numerical results align with the result of Theorem 1.
We take results of APTER using univariate experts, so that $m=2d$.
The "c-index error" ($C_{\mbox{err}}$) is given for different values of $d$ and $n$.
$C_{\mbox{err}}$ is computed as the difference between the $C_n$ obtained by APTER - denoted as $\hat{f}$ - and the $C_n$ of the single "best" expert $f_j(\x_i)=\x_{i,j}$:
\begin{equation}
	C_{\mbox{err}}=\max_j C_n(f_j) - C_n\left(\hat{f}\right).
\end{equation}
This formula is similar to equation (\ref{eq.loss2}).
The numerical performances are displayed in Figure (\ref{fig:cerror}.c).
This figure indicates that $C_{\mbox{err}}$ increases logarithmically in $d$, and in terms of $\frac{1}{\sqrt{n}}$. 
This supports the result of Theorem 1.

\begin{figure}[htbp] 
   \centering
   \subfigure[]{\includegraphics[width=2.1in]{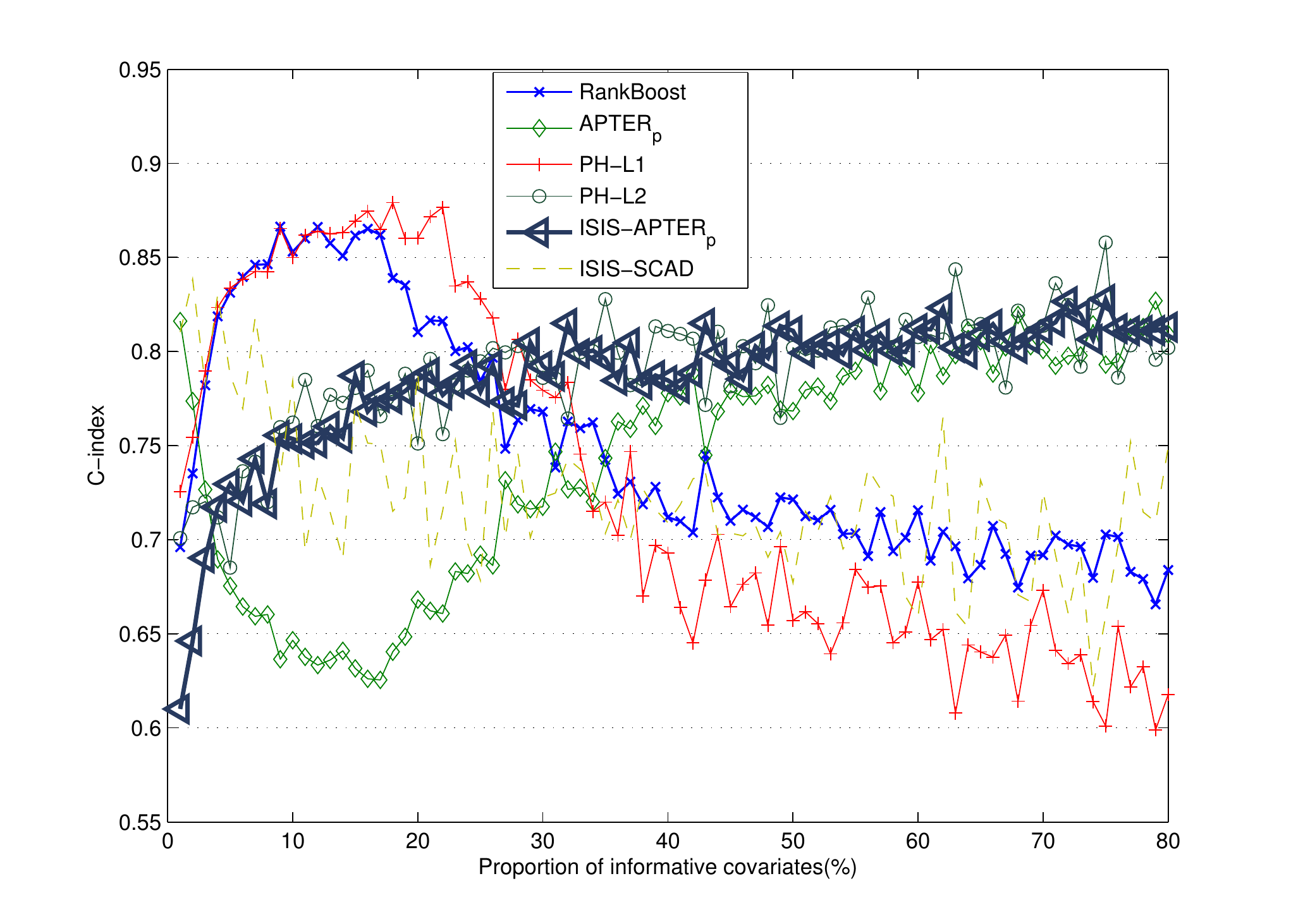}} 
   \subfigure[]{\includegraphics[width=2.1in]{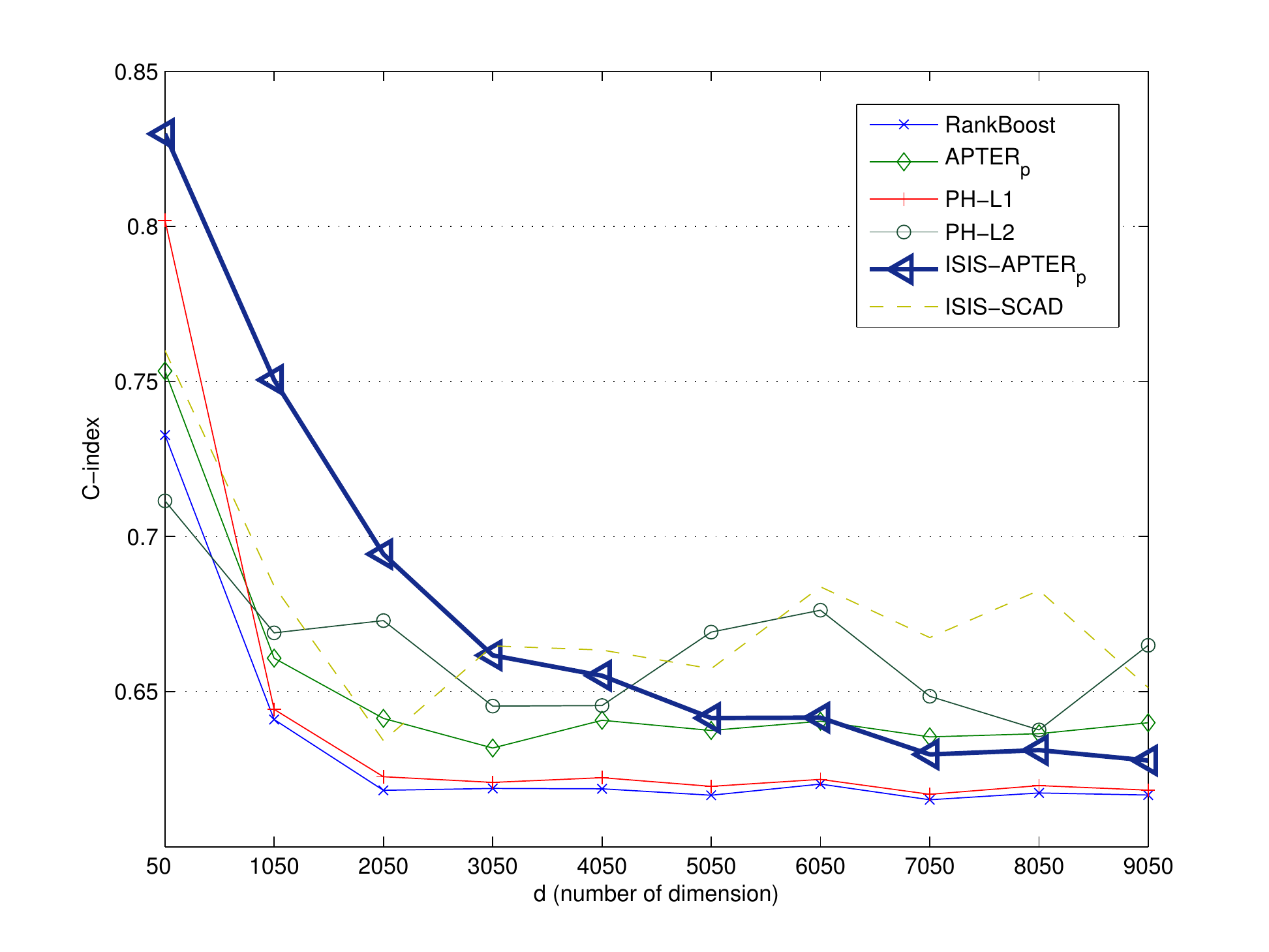}} 
      \subfigure[]{\includegraphics[width=2.1in]{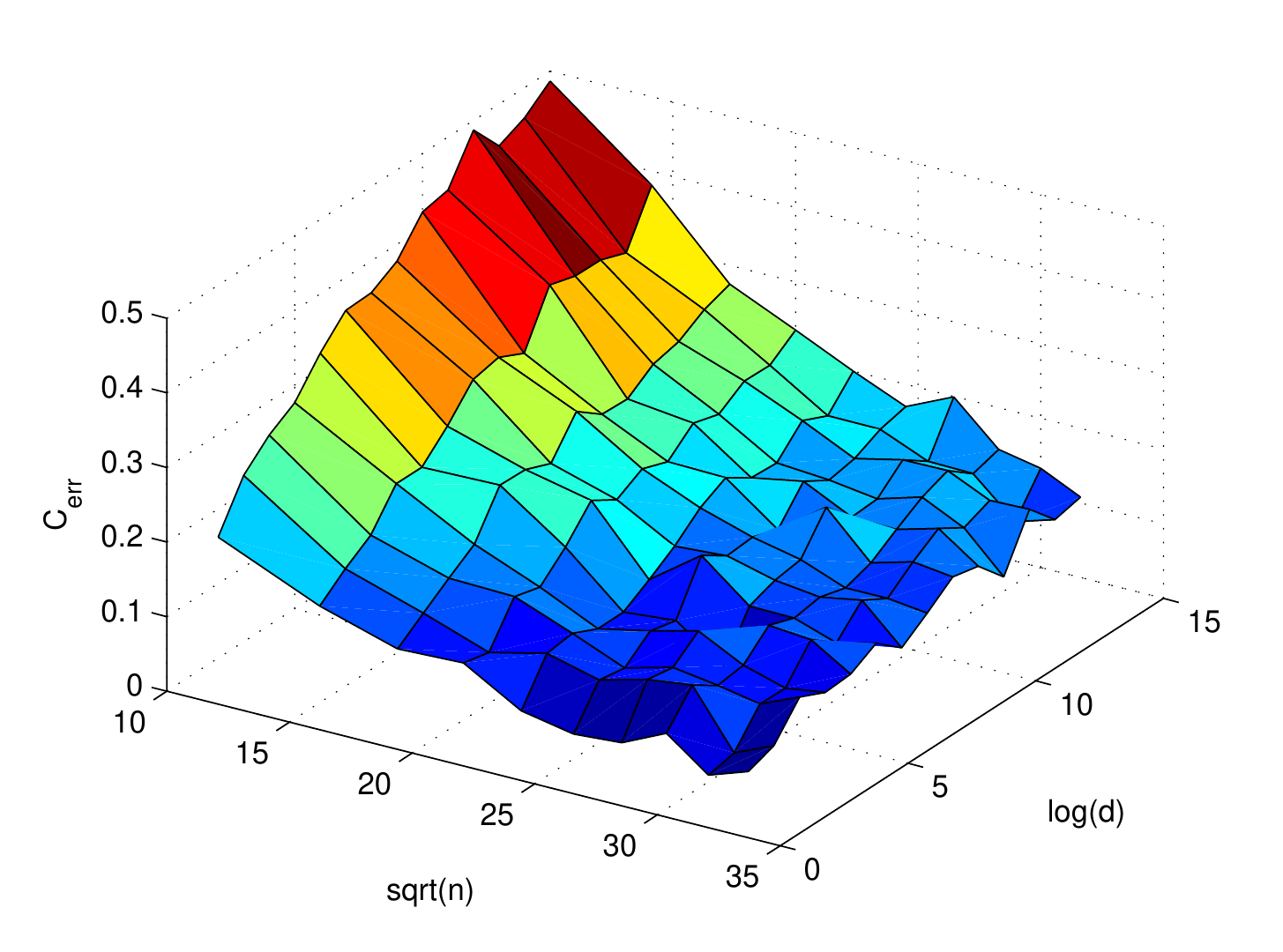}}  \\
\caption{
   	Comparison of the numerical results obtained on the artificial data sets
	(a) when keeping $d=100$ fixed, and 
	(b) when keeping $k=10$ fixed. 
   	(c) The evolution of the 'C-index error' $C_{\mbox{err}}$ obtained by APTER$_p$ for different values of $(n,m)$.
   }
   \label{fig:cerror}
\end{figure}

\subsection{Real Datasets}

In order to benchmark APTER and its variations against state-of-the-art approaches,
we run the algorithms on a wide range of large-dimensional real datasets.
The dataset are collected in a context of bioinformatics, and a full description of this data can be found on the website.
The experiments are divided into three categories:

(i) The algorithms are run on seven micro-array datasets, in order to asses performance on typical sizes for those datasets.
	Here we see that there is no clear overall winner amongst the algorithms, but the proposed algorithm (ISIS-APTER$_p$)
	does do repeatedly very well, and performs best on most datasets.
	Results are given in Table (\ref{tab:exp}).

(ii) In order to see wether the positive performance is not due to irregularities of the data,
	we consider the following {\em null} experiment. Consider the AML dataset, but lets shuffle the observed phenotypes 
	(the observed $Y$) between different subjects. So any relation between the expression level and the random phenotype must be due 
	to plain chance (by construction). We see in Figure (\ref{fig:naml},a) that indeed the distribution of the methods based on this {\em shuffled} 
	data nears a neutral $C_n$ on the test set of $0.50$. This means that the 10\% improvement as found in the real experiment (see table)
	is substantial with respect to the randomizations, and are not due to chance alone.

(iii) The results of the algorithm is compared on the micro-array dataset as reported in \cite{dave2004prediction}, and analysed further in \cite{tib2006reanalysis}.
	Here we found that the obtained performance is significantly larger than what was reported earlier, while we do not have to resort to the 
	{\em clustering} preprocessing as advocated in \cite{dave2004prediction,tib2006reanalysis}.
	This data has a very high dimensionality ($d=44.928$) and has only a few cases ($n=191$). 
	Results are given in Table (\ref{tab:fl}) and the box plots of the performances due to the 50 randomisations, are given in Figure (\ref{fig:naml}.b).

\begin{figure}[htbp] 
   \centering
   \subfigure[]{\includegraphics[width=3in]{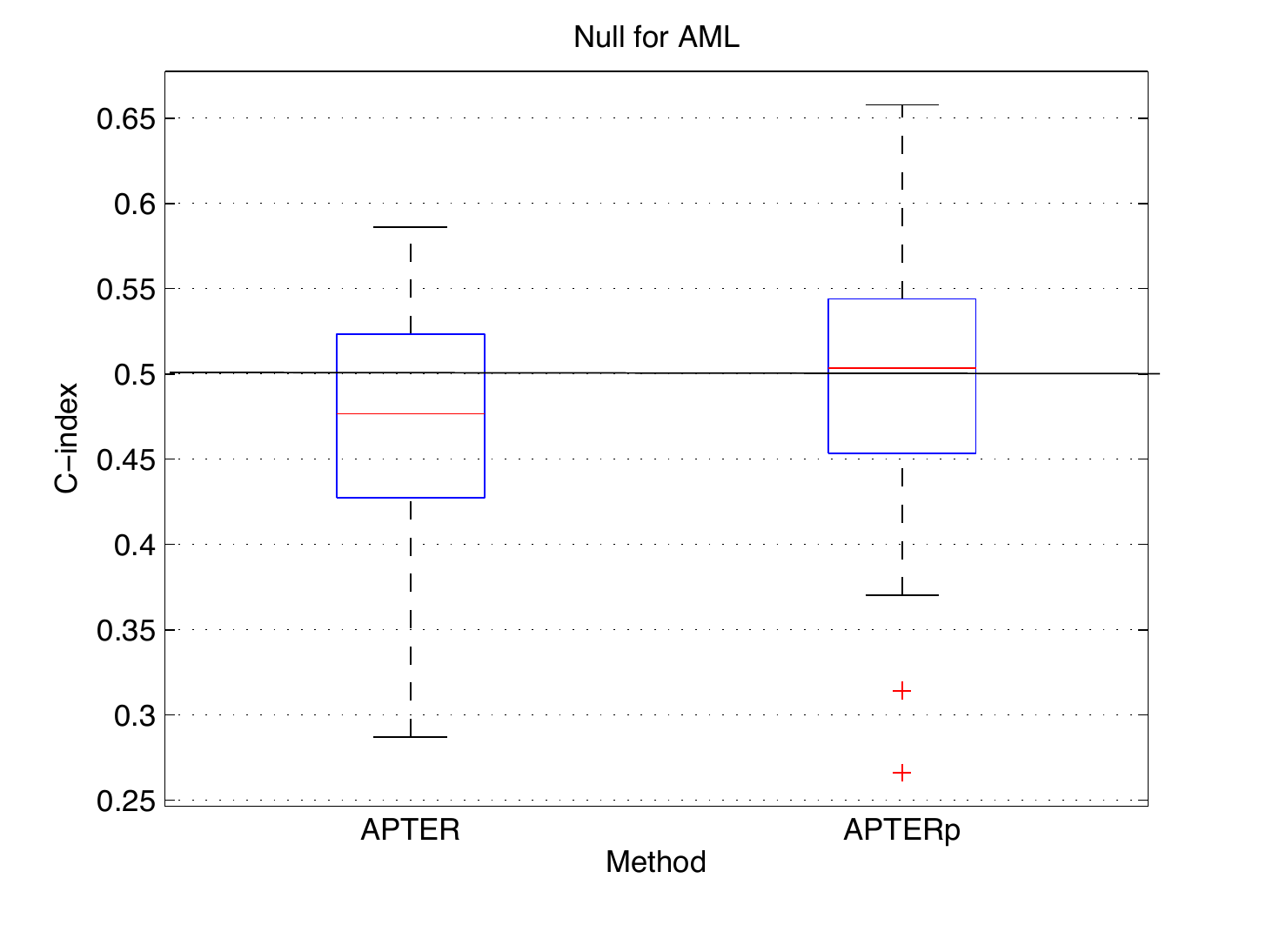}}
      \subfigure[]{\includegraphics[width=3in]{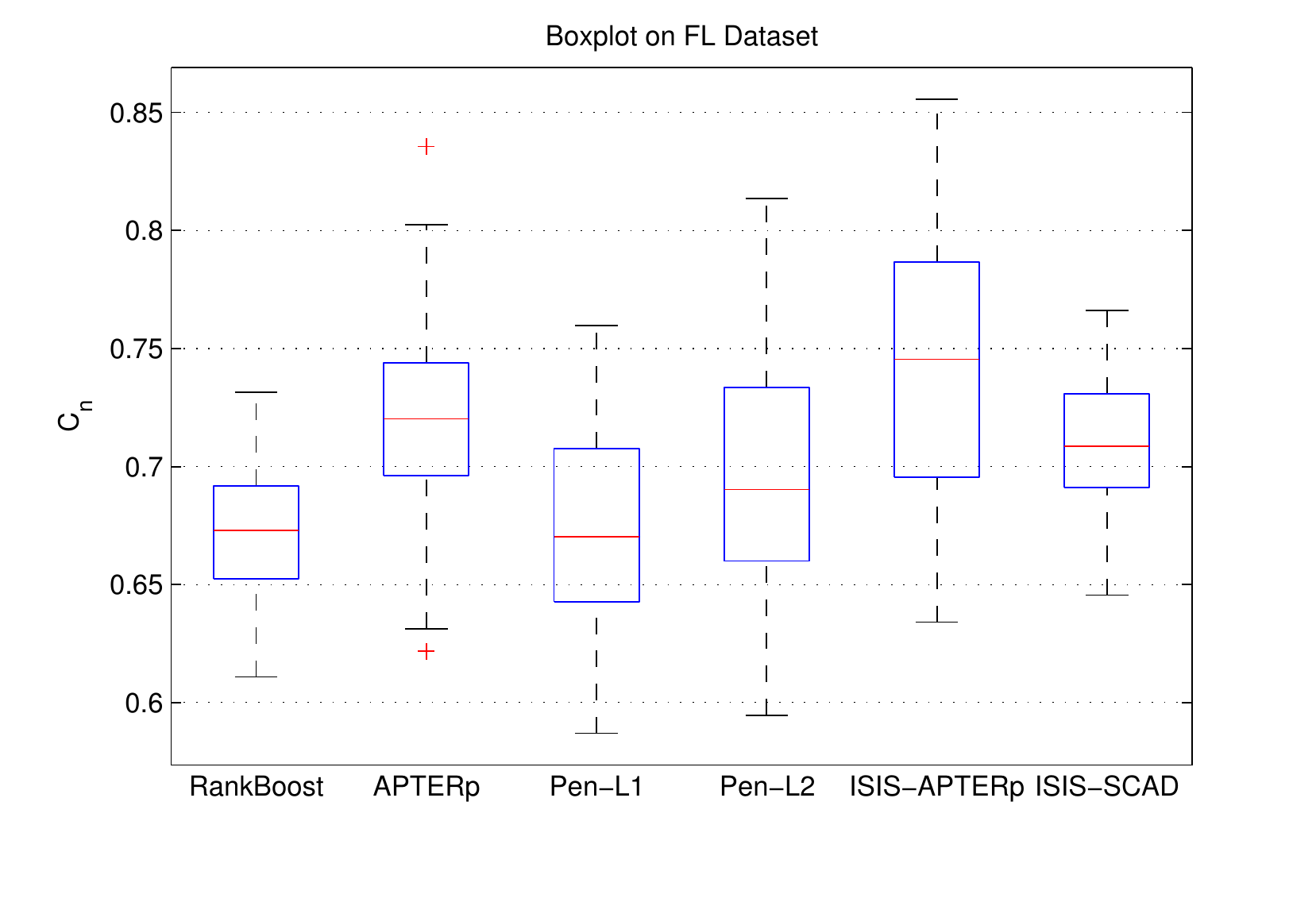}}
   \caption{	(a) Performances of APTER and APTER$_p$ on the AML dataset when the reposes are randomly shuffled.
   		(b)    	Boxplots of the numerical results obtained on the FL dataset. 
   	Results are expressed in terms of the $C_n(f)$ on a test set, where $f$ is trained and tuned on a disjunct training set.
	The boxplots are obtained using 50 randomizations of the split training-testset.}
   \label{fig:naml}
\end{figure}

Finally, we discuss the application of the method on the same high-dimensional ($d=44.928$) dataset as before,
but we study the impact of the parameter $m$ given to ISIS, which returns in turn the data to be processed by APTER$_p$.
The performances for different values of $m$ are given in Fig. (\ref{fig:fl}.a).
The best performance is achieved for $m=800$, which is the value which was used in the earlier experiment reported in Fig. (\ref{fig:naml}.b).
Here we compare only to a few other approaches, namely the PH-L$_1$, MINLIP$_p$ and MODEL2 approach which make all use of an explicit optimisation scheme.
Panel (\ref{fig:fl}.b) reports the time needed to perform training/ tuning and randomisation corresponding to a fixed value of $m$.
Panel (\ref{fig:fl}.c) reports the size of the memory used up for the same procedure.
Here it is clearly seen that  APTER$_p$ results in surprisingly good performance, given that it uses up less computations and memory.
It is even so that the optimisation-based techniques cannot finish for large $m$ in reasonable time or without problems of the memory management, 
despite the fact that a very efficient optimisation solver (Yalmip) was used to implement those.

\begin{figure}[htbp] 
   \centering
   \subfigure[]{\includegraphics[width=3.1in]{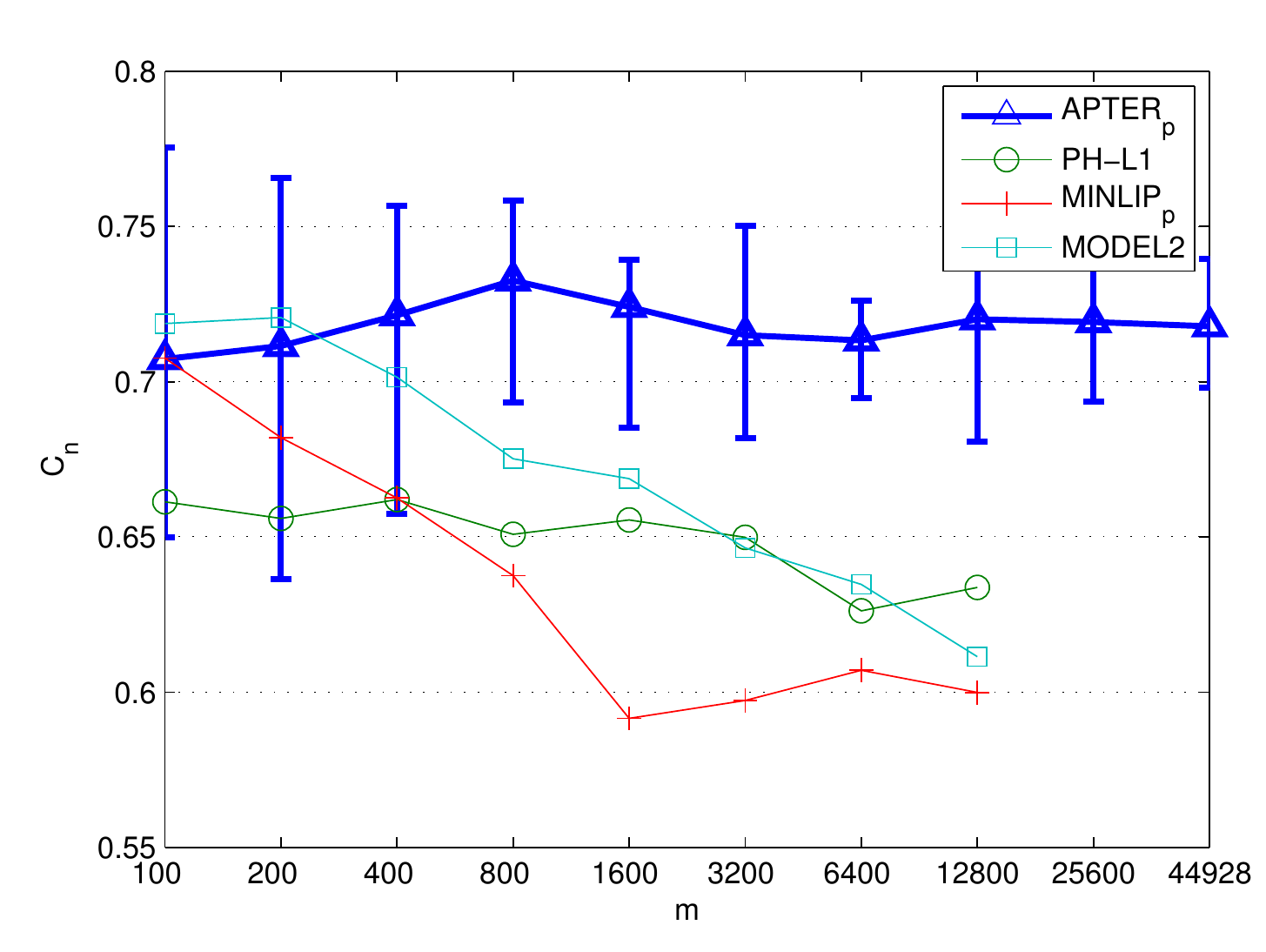}}
   \subfigure[]{\includegraphics[width=3.1in]{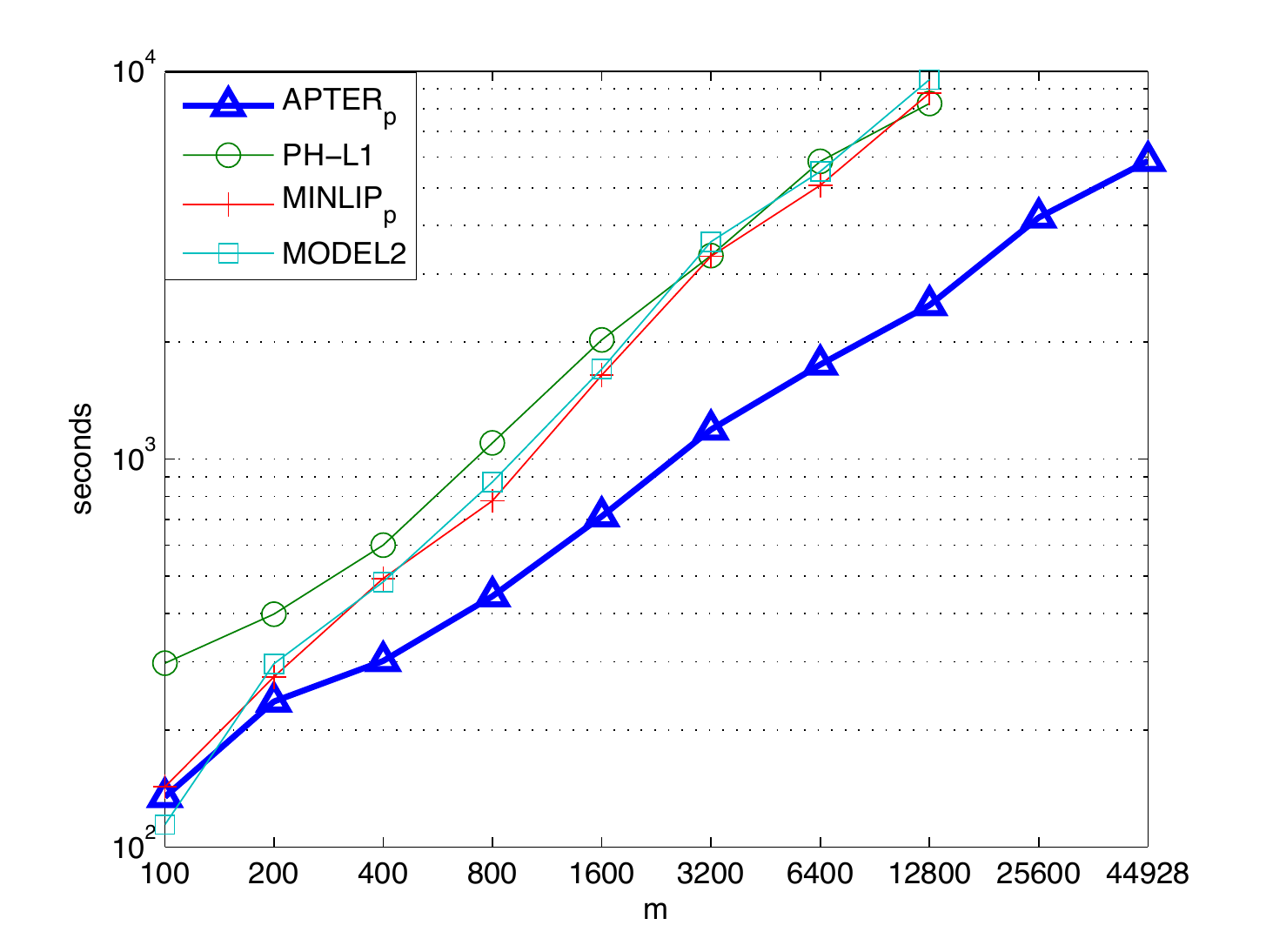}}
   \subfigure[]{\includegraphics[width=3.1in]{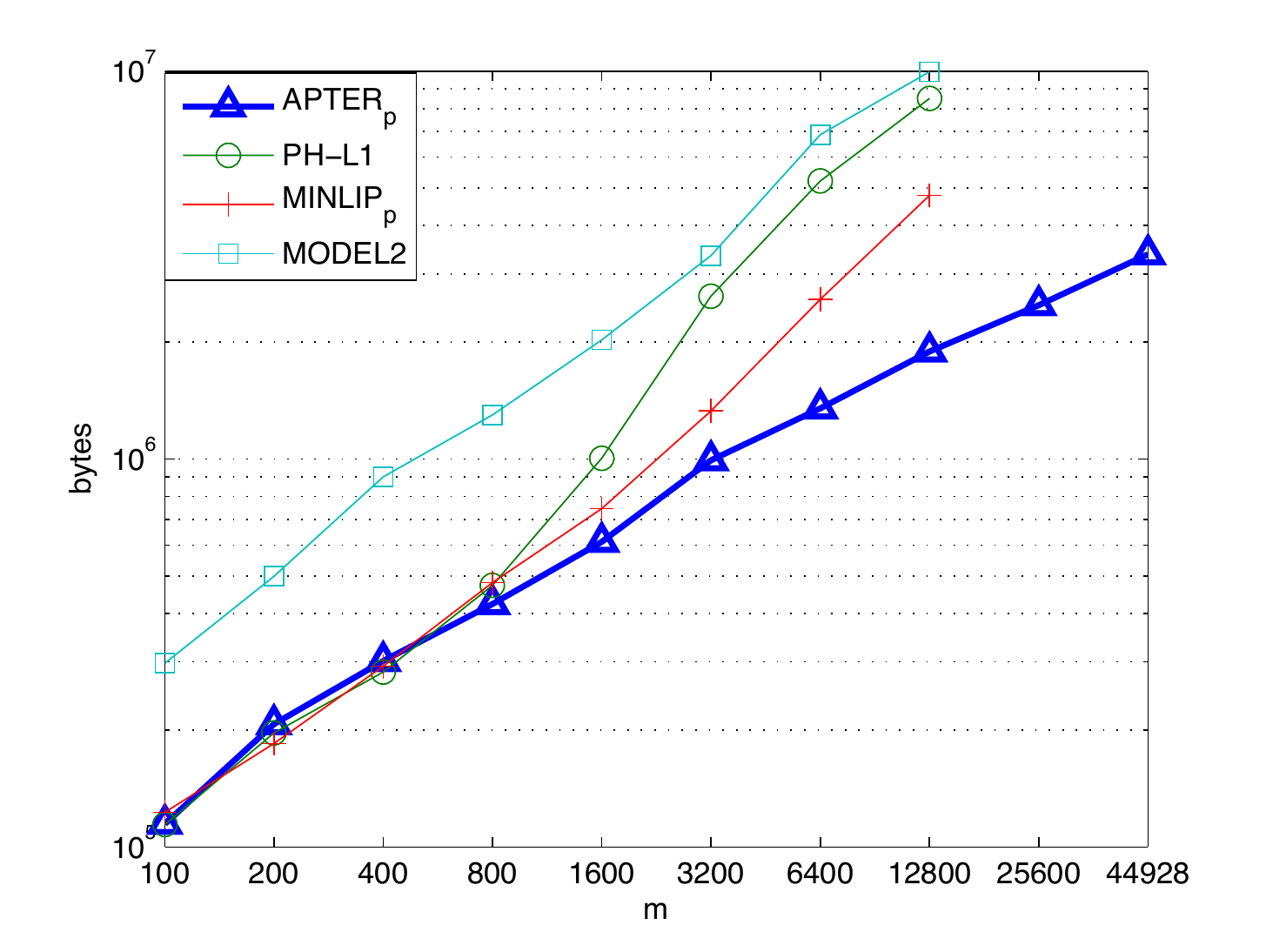}}
   \caption{results of the choice of $m$ in ISIS, based on the Follicular Lymphoma dataset\cite{dave2004prediction,tib2006reanalysis}.
   	(a) Performance expressed as $C_n(\hat{f})$ on the test sets (medium ($\pm$) of 50 randomizations).
	(b) Computation time for running tuning, training and randomisation for a fixed value of $m$.
	(c) Usage of memory of the same procedure.
   }
   \label{fig:fl}
\end{figure}

\subsection{Discussion of the Results}

These results uncover some interesting properties of the application of the proposed algorithms in this bio-informatics setting. 

First of all, the APTER and APTER$_p$ methods are orders of magnitudes faster (computationally) compared to the bulk of methods based on 
optimization formulations (either using Maximum (penalized) Partial Likelihood, Empirical Risk Minimization or multivariate preprocessing techniques).
This does not affect the performance in any way, contrary to what intuition would suggest.
In fact, the performance on typical micro-array data of the {\em vanilla} APTER or APTER$_p$ (without ISIS) is often amongst the better.

Secondly, inclusion of preprocessing with ISIS - also very attractive from a computational perspective - is boosting up significantly the 
performance of APTER. We have no theoretical explanation for this, since ISIS was designed to complement $L_1$ or Danzig-selector approaches. 
While the authors of ISIS advocate the used of a SCAD-based approach based on empirical evidence, 
we find that APTER$_p$ is overall a better choice for the mentioned datasets.
It becomes clear by looking for example to the results on the FL dataset, that ISIS per se is not causing the boost in performance.
However, the combination of ISIS and APTER$_p$ seems to perform surprisingly well.

Furthermore, the empirical results indicate that the statistical performance is preserved by using APTER$_p$ combined with ISIS, and may even improve over
performances obtained using existing approaches. This is remarkable since the computational power is orders of magnitude smaller than most existing approaches 
based on (penalised) PL of ERM. We find also that empirical results align quite closely the theoretical findings as illustrated with an experiment on artificial data.
\footnote{As is common for such case-studies, there is considerable uncertainty (variability) of the results (see e.g. the box plots in Fig. (\ref{fig:naml})).
However, since results are calculated on independent test-sets, this does not indicate overfitting. Note that this is supported by the theory indicating that the technique can deal with large sets of covariates without overfitting on the data.}


\section{Conclusions}

This paper presents statistically and computationally compelling arguments for a method based on 
aggregation can be used for analysis of survival data in high dimensions.
Theoretical findings are complemented with empirical results on micro-array datasets.
We feel that this result is surprising not only in that it outperforms 
methods in ERM or (penalised) PL, but provides as well a tool with much lower computational
complexity as the former ones since no direct optimization is involved.
We present empirical, reproducible results which support this claim of efficiency.
This analysis presents many new opportunities, both applied (towards Genome Wide Analysis, or GWAs)
as well as theoretical (can we improve the rates of convergence by choosing other loss functions?).

\appendix
\section{Proof of Theorem 1}

The following results will be used.
\begin{Lemma}[Hoeffding]
	Let $\lambda\in\R$, and let $X$ be a random variable taking values in $[a,b]\subset\R$, then
	\begin{equation}
		\ln \E \left[ \exp(\lambda X)\right] \leq \lambda \E[X] + \frac{\lambda^2(b-a)^2}{8}.
		\label{eq.hoeff}
	\end{equation}
\end{Lemma}
With assumption of eq. (\ref{eq.cal}) in hand, the following result holds:
\begin{Lemma}
	Given $m$ experts $\{f_i:\R_d\rightarrow\R\}_{i=1}^m$, a loss function $\ell:\R\rightarrow\R$ satisfying eq. (\ref{eq.cal}),
	and let $\{(\x_k,Y_k, \delta_k)\}_{k=1}^n$ take values in $\R^d\times \R_+\times\{0,1\}$.
	Let the APTER algorithm (\ref{alg.apter}) be run with a fixed $\nu>0$, then
	\begin{equation}
		\E_{n-1}\left[
			\mathbb{L} (\hat\p)
			- 
			\min_{i=1,\dots,m} \mathbb{L}(f_i)   
		\right]
		\leq 
		\frac{\ln m}{\nu n} + \frac{1}{\nu}\E_n [R_n],
		\label{eq.loss}
	\end{equation}
	with
	\begin{equation}
		 R_n = \frac{1}{\nu} \sum_{t=1}^n \ln \hat{E} \exp -\nu \left(\ell_n(f) - \hat{E} \ell_n(f)\right).
		\label{eq.Rn}	
	\end{equation}
	and $\hat{E} g(f) = \sum_{i=1}^m \hat\p_i g(f_i)$ for any $g$.
\end{Lemma}
\begin{proof}
	Consider the evolution of the normalization terms $W_t$ where
	\begin{equation}
		W_t = \sum_{i=1}^m \exp -\nu L_t(f_i),
		\label{eq.thm2.1}
	\end{equation}
	is characterized. 
	Specifically, we see that  
	\begin{equation}
		\ln\frac{W_n}{W_0} 
		= \ln \sum_{i=1}^m \exp(-\nu L_n (f_i)) - \ln m \\
		\geq -\nu \min_{i=1,\dots, m} L_n(f_i) - \ln m,
		\label{eq.thm2.2}
	\end{equation}
	as before. Hence
	\begin{multline}
		\frac{1}{n\nu}\E_n\left[ \ln\frac{W_n}{W_0} \right]
		\geq -\min_{i=1,\dots, m} \E_n\left[ \frac{1}{n}L_n(f_i)\right] - \frac{\ln m}{n\nu} \\
		\geq -\min_{i=1,\dots, m} \E_n\left[ \ell_n(f_i)\right] - \frac{\ln m}{n\nu} \\
		\geq -\min_{i=1,\dots, m} \E_{n-1} \mathbb{L}(f_i) - \frac{\ln m}{n\nu}.
		\label{eq.thm2.2b}
	\end{multline}
	On the other hand we have that
	\begin{equation}
		\ln\frac{W_t}{W_{t-1}} 
		= \ln \frac{\sum_{i=1}^m \exp(-\nu L_t(f_i))}
		{\sum_{j=1}^m \exp(-\nu L_{t-1}(f_j))} \\
		= \ln \sum_{i=1}^m \p_i^{t-1}\left( \exp -\nu \ell_t(f_i)\right).
		\label{eq.thm2.3}
	\end{equation}
	Taking expectation over the $n$ samples (denoted as $\E_n[\cdot]$) seen thus far, and summarizing over $t=1,\dots,n$ gives
	\begin{multline}
		\frac{1}{n\nu}\sum_{t=1}^n \E_n\left[ \ln W_t -  \ln W_{t-1}\right] \\
		= \frac{1}{n\nu} \sum_{t=1}^n \E_n\left[\ln \sum_{i=1}^m \p_i^{t-1} \exp -\nu \ell_t(f_i)\right] \\
		= \frac{1}{n\nu} \sum_{t=1}^n \E_n\left[\ln \sum_{i=1}^m \p_i^{t-1} \exp -\nu \frac{L_n(f_i)}{n}\right] \\
		= \frac{1}{n\nu} \sum_{t=1}^n \E_n\left[\ln \sum_{i=1}^m \p_i^{t-1} \exp -\nu \ell_n(f_i)\right] \\
		\leq \frac{1}{\nu}  \E_n\left[\ln \sum_{i=1}^m \hat\p_i \exp -\nu \ell_n(f_i)\right],
		\label{eq.thm2.4}
	\end{multline}
	where the last inequality follows from Jenssen's inequality, and from the formula of aggregation as in eq. (\ref{eq.aggregation}).
	Now, this gives
	\begin{multline}
		\frac{1}{\nu} \E_n\left[\ln \hat{E} \exp -\nu \ell_n(f)\right] \\
		= \frac{1}{\nu} \E_n \left[\ln \hat{E} \exp -\nu \hat{E}\ell_n(f)\right] \\
		+ \frac{1}{\nu} \E_n\left[\ln \hat{E} \exp -\nu \left(\ell_n(f) - \hat{E}\ell_n(f)\right)\right] \\
		= -\E_{n-1} \E^n[\hat{E}\ell_n(f)] \\
		+ \frac{1}{\nu} \E_n\left[\ln \hat{E} \exp -\nu \left(\ell_n(f) - \hat{E}\ell_n(f)\right)\right],
		\label{eq.thm2.5}
	\end{multline}
	where we defined for notational convenience $\hat{E}\x = \sum_{i=1}^m \hat\p_i \x_i$ for all $\x\in\R^m$, and 
	$\hat{E} \ell_n(f) = \sum_{i=1}^m \hat\p_i \ell_n(f_i)$.
	Combining inequalities (\ref{eq.thm2.2b}) and (\ref{eq.thm2.5}) gives
	\begin{equation}
		 \E_{n-1}\left[
		 \mathbb{L} (\hat\p) - \min_{i=1,\dots, m} \mathbb{L}(f_i) 
		 \right]
		\leq 
		\frac{\ln m}{\nu n} + \frac{1}{\nu}\E_n [R_n],
		 \label{eq.thm2.6}
	\end{equation}
	as desired. $\Box$
\end{proof}
So we are left to proof that the term $\E_n[R_n]$ is bounded in our case.
The proof of Theorem 1 is then given as follows.
\begin{proof}
	This follows by application of Hoeffding's inequality as in eq. (\ref{eq.hoeff}) since 
	\begin{equation}
		R_n = \ln \hat{E} \exp -\nu \left(\ell_n(f) - \hat{E}\ell_n(f)\right) \\
			\leq \frac{\nu^2}{2},
		\label{eq.thm3.1}
	\end{equation}
	where we use that $0\leq \ell_n\leq 1$.
	Then combining with eq.  (\ref{eq.loss}) gives the result. $\Box$
\end{proof}

\section{Benchmark datasets}

This appendix describes the real-world datasets.
The datasets range from large-dimensional ($d=O(100)$) to huge-dimensional ($d=O(10,000)$), 
and record $n=O(100)$ subjects. 
We report the performance of different methods on:
\begin{itemize}
\item 7 publicly available datasets containing micro-array expression levels and events (occurrence of disease) of the associated subjects as used in \cite{bovelstad2007predicting}.
\item The micro-array survival dataset as presented in \cite{dave2004prediction} and analysed in the report \cite{tib2006reanalysis}.
\end{itemize}
Details are given below.
The 7 publicly available micro-array datasets as used for benchmarking in \cite{bovelstad2007predicting}, have the following properties. 
\begin{itemize}
\item [(NSBCD):] 
	The Norway/Stanford Breast Cancer Data set is given in \cite{sorlie2003repeated}.
	In this database there are survival data of $n=115$ women who have breast cancer, 
	and $d=549$ intrinsic genes introduced in \cite{sorlie2003repeated} were measured.
	In the $115$ patients, 33\% (38) have experienced an event during the study. 
	Missing values were imputed by the 10-nearest neighbour method.

\item[(Veer):] 
	The survival data of sporadic lymph-node-negative patients with their gene expression profiles is given in \cite{van2002gene}.
	It has $n=78$ patients with $d=4751$ gene expressions selected from the 25,000 genes recorded with the micro-array.
	44 patients remained free of disease after their diagnosis for an interval of at least 5 years. 
	The average follow-up time for these patients was 8.7 years.
	34 patients had developed distant metastases within 5 years, and the mean time to metastases was 2.5 years.

\item[(Vijver):] 
	The data set of $n=295$ consecutive patients with primary breast carcinomas is from \cite{van2002gene}
	All patients had stage I or II breast cancer and were younger than 53 years old.
	They gave the previously determined $d=70$ marker genes that are associated with 
	the risk of early distant metastases in young patients with lymph-node-negative breast cancer.
	The median follow-up among all 295 patients was 6.7 years (range, 0.05 to 18.3). There were no missing data.
	88 patients have experienced an event during the study.

\item[(DBCD):] 
	The Dutch Breast Cancer Data set is described in \cite{van2006cross}, 
	and is a subset of the data from \cite{van2002gene}.
	There are survival data of $n=295$ women who have breast cancer.
	The measures of $d=4919$ gene expression were taken from the fresh-frozen-tissue bank of the Netherlands Cancer Institute.
	All the ages of the patients are smaller than or equal to 52 years. 
	The diagnosis was made between 1984 and 1995 without previous history of cancer. 
	The median of follow-up time was 6.7 years (range 0.05-18.3).
	In the 295 patients, 26.78\% (79) have experienced an event during the study.
	
\item[(DLBCL):] 
	The diffuse large-B-cell lymphoma data set is described in \cite{rosenwald2002use}.
	This contains survival data of $n=240$ patients who have diffuse large-B-cell lymphoma.
	$d=7399$ different gene expression measurements are given.
	The median of follow-up time was 2.8 years.
	From the 240 patients, 58\% have experienced an event during the study.

\item[(Beer):] 
	The survival data of $n=86$ patients with primary lung adenocarcinomas is from \cite{beer2002gene}
	There are $d=7129$ expressed genes selected from Affymetrix hu6800 micro-arrays.
	76 patients have experienced an event during the study.

\item[(AML):] 
	The survival data of acute myeloid leukemia patients is described in \cite{bullinger2004use}.
	It contains $n=116$ patients with acute myeloid leukemia and the expression levels of $d=6283$ genes.
	71 patients have experienced an event during the study.
\end{itemize}
The same datasets were used in \cite{bovelstad2007predicting} and \cite{van2011learning} 
to benchmark state-of-art methods, results that are reproduced here as well.
\begin{itemize}
\item[(FL):]
	Additionally, we use the micro-array dataset which was used in \cite{dave2004prediction}, and analysed in \cite{tib2006reanalysis}.
	This dataset included the survival data of $n=191$ patients with follicular lymphomas, where $t_0$ equals the respective time of diagnosis.
	The median age at diagnosis  was 51 years (range, 23 to 81), and the median follow-up time was 6.6 years (range, less than 1.0 to 28.2).
	The median followup time among the patients alive at the final follow-up was 8.1 years.
	It contains $d=44928$ gene expression levels selected from Affymetrix U133A and U133B micro-arrays.
\end{itemize}


\newpage

\begin{table}
\begin{tabular}{|c|c|c|c|c|c|c|c|}
\hline 
 & NSBCD & DBCD & DLBCD & Veer & Vijver & Beer & AML \tabularnewline
 & $(115\times 549)$& $(295\times 4919)$& $(240\times 7399)$& $(78\times 4751)$& $(295\times 70)$& $(86\times 7129)$& $(116\times 6283)$ \tabularnewline
\hline 
\hline 
APTER & 0.73$\pm$0.10 & 0.69$\pm$0.06 & 0.58$\pm$0.04 & 0.65$\pm$0.10 & 0.44$\pm$0.06 & 0.60$\pm$0.13 & 0.58$\pm$0.05\tabularnewline
\hline 
APTER$_p$ & 0.77$\pm$0.05 & 0.74$\pm$0.04 & 0.59$\pm$0.03 & \textbf{0.68$\pm$0.08} & 0.62$\pm$0.04 & 0.73$\pm$0.08 & 0.60$\pm$0.05\tabularnewline
\hline 
MINLIP$_p$ & 0.74$\pm$0.05 & 0.71$\pm$0.04 & 0.59$\pm$0.04 & 0.65$\pm$0.10 & 0.61$\pm$0.06 & 0.69$\pm$0.09 & 0.55$\pm$0.07\tabularnewline
\hline 
MODEL2 & 0.75$\pm$0.04 & 0.74$\pm$0.04 & 0.62$\pm$0.03 & \textbf{0.67$\pm$0.09} & 0.61$\pm$0.06 & \textbf{0.74$\pm$0.08} & 0.56$\pm$0.06\tabularnewline
\hline 
PLS & \textbf{0.78$\pm$0.05}  & 0.74$\pm$0.03 & 0.53$\pm$0.05 & 0.58$\pm$0.10 & 0.62$\pm$0.07 & 0.66$\pm$0.12 & 0.57$\pm$0.06\tabularnewline
\hline 
PH-L2 & 0.69$\pm$0.07 & 0.73$\pm$0.04 & \textbf{0.65$\pm$0.04} & 0.64$\pm$0.08 & 0.61$\pm$0.08 & 0.73$\pm$0.08 & 0.54$\pm$0.06\tabularnewline
\hline 
PH-L1 & 0.69$\pm$0.06 & 0.74$\pm$0.04 & 0.60$\pm$0.04 & 0.60$\pm$0.06 & \textbf{0.65$\pm$0.06} & 0.69$\pm$0.02 & 0.61$\pm$0.06\tabularnewline
\hline 
Rankboost & 0.75$\pm$0.04 & 0.72$\pm$0.03 & 0.62$\pm$0.02 & 0.62$\pm$0.02 & \textbf{0.65$\pm$0.02} & 0.71$\pm$0.02 & 0.53$\pm$0.01\tabularnewline
\hline 
ISIS-SCAD & 0.69$\pm$0.04 & 0.72$\pm$0.04 & \textbf{0.65$\pm$0.07} & \textbf{0.68$\pm$0.04} & 0.62$\pm$0.02 & 0.72$\pm$0.04 & \textbf{0.63$\pm$0.02}\tabularnewline
\hline 
ISIS-APTER$_p$ & \textbf{0.78$\pm$0.06} & \textbf{0.76$\pm$0.08} & 0.62$\pm$0.07 & 0.66$\pm$0.05 & 0.62$\pm$0.06 & \textbf{0.75$\pm$0.09} & 0.59$\pm$0.05\tabularnewline
\hline 
\end{tabular}
  \caption{This table reports the $C_n(\hat{f})$ as computed on an independent test set (median $\pm$ variance on 50 randomisations)
  	of the experiments of 10 different methods on 7 micro-array datasets.}
  \label{tab:exp}
\end{table}

\begin{table}{
  \begin{tabular}{|c|c|c|} 
     \toprule
     Dataset & Method & $C_n(\hat{f})$  \\
     \midrule
FL & APTER 			& 0.70$\pm$0.05 	\\
$(191\times 44928)$ 	& APTER$_p$ 	& 0.73$\pm$0.04  \\
					& MINLIP$_p$ 	& 0.70$\pm$0.03 	 \\
					& MODEL2  	& 0.72$\pm$0.04 	 \\
					& PLS 		& 0.66$\pm$0.03 	\\
					& PH-L$_2$ 	& 0.69$\pm$0.07 	\\
					& PH-L$_1$ & 0.67$\pm$0.05 		 \\
					& RankBoost &0.67$\pm$0.03 	  \\
					& ISIS-SCAD & 0.71$\pm$0.03 	 \\
					& ISIS-APTER$_p$ & {\bf 0.74$\pm$0.05} 	 \\
					& Dave's Method (see \cite{tib2006reanalysis}) & 0.71$\pm$0.02 	 \\
     \bottomrule
  \end{tabular}}
  \caption{Numerical results of the experiments of 10 different methods on the Follicular Lymphoma dataset.}
  \label{tab:fl}
\end{table}


\end{document}